\documentclass[onecolumn]{article}

\usepackage{tabulary,graphicx,times,caption,fancyhdr,amsfonts,amssymb,amsbsy,latexsym,amsmath}
\usepackage[utf8]{inputenc}
\usepackage{url,multirow,morefloats,floatflt,cancel,tfrupee,textcomp,colortbl,xcolor,pifont}
\usepackage[nointegrals]{wasysym}
\makeatletter

\usepackage{ifxetex}
\ifxetex\else
\usepackage{dblfloatfix}
\fi

\@ifundefined{subparagraph}{
	\def\subparagraph{\@startsection{paragraph}{5}{2\parindent}{0ex plus 0.1ex minus 0.1ex}%
		{0ex}{\normalfont\small\itshape}}%
}{}

\def\URL#1#2{\@ifundefined{href}{#2}{\href{#1}{#2}}}

\def\UrlOrds{\do\*\do\-\do\~\do\'\do\"\do\-}%
\g@addto@macro{\UrlBreaks}{\UrlOrds}

\makeatother

\usepackage[paperheight=11in,paperwidth=8.3in,margin=2.5cm,headsep=.7cm,top=2.5cm]{geometry}
\usepackage[T1]{fontenc}

\renewenvironment{abstract}
{\trivlist\item[]\leftskip0pt\par\vskip4pt\noindent
	\textbf{\abstractname}\mbox{\null}\\}
{\par\noindent\endtrivlist}

\def\keywords#1{\par\medskip\par\noindent\textbf{Keywords}: #1\par}

 \date{} \emergencystretch 8pt

\makeatletter
\def\author#1{\gdef\@author{\hskip-\tabcolsep%
		\parbox{\textwidth}{\raggedright\bfseries#1\\[1pc]}}}
\def\address[#1]#2{\g@addto@macro\@author{\\\hskip-\tabcolsep\parbox{\textwidth}{\raggedright%
			\normalsize\normalfont\textsuperscript{#1}#2}}}

\def\correspondence#1{\g@addto@macro\@author{\\\hskip-\tabcolsep\parbox{\textwidth}{\raggedright%
			\vspace*{10pt}\normalsize\normalfont~\\#1~\\[12pt]}}}
\def\email#1{\g@addto@macro\@author{\\\hskip-\tabcolsep\parbox{\textwidth}{\raggedright%
			\normalsize\normalfont Emails: #1}}}

\def\title#1{\gdef\@title{\vspace*{-30pt}%
		\raggedright\textbf{\@journaltitle}~\\%
		\raggedright\bfseries\ifx\@articleType\@empty\vspace*{20pt}\else%
		\vspace*{20pt}\@articleType\vspace*{20pt}\\\fi#1}}
\let\@journaltitle\@empty \def\journaltitle#1{\gdef\@journaltitle{{\normalfont\itshape#1}}}
\let\@articleType\@empty \def\articletype#1{\gdef\@articleType{{\normalfont\itshape#1}}}

\let\@runningHead\@empty \def\RunningHead#1{\gdef\@runningHead{{\normalfont #1}}}

\usepackage{fancyhdr}
\fancypagestyle{headings}{\fancyhf{}
	\fancyhead[R]{\itshape\@runningHead}
	\fancyfoot[C]{\thepage}}
\fancypagestyle{plain}{%
	\fancyhf{}\fancyhead[R]{\LaTeX\ template for Hindawi journal articles}
	\fancyfoot[C]{\thepage}}
\makeatother

\usepackage[%
numbers,sort&compress%
]{natbib}

\usepackage{float,xcolor}
\usepackage{multicol}
\allowdisplaybreaks[4]
\usepackage[title]{appendix}
\usepackage[switch]{lineno}

\journaltitle{Complexity}
\articletype{Research Article} 

\usepackage{ntheorem}
\newtheorem{thm}{\bf Theorem}
\newtheorem*{proof}{Proof}

\begin{document}
	\title{A general epidemic model and its application to mask design considering different preferences towards masks}
	
	\author{%
		Chaoqian Wang and Hamdi Kavak
	}
	
	\address[]{Department of Computational and Data Sciences, George Mason University, Fairfax, VA 22030, USA}

	\correspondence{}
	\email{Chaoqian Wang <CqWang814921147@outlook.com>, <cwang50@gmu.edu>; Hamdi Kavak <hkavak@gmu.edu>}

	\RunningHead{Running head}
	
	\maketitle 
	
	\begin{abstract}
		While most masks have a limited effect on personal protection, how effective are they for collective protection? How to enlighten the design of masks from the perspective of collective dynamics? In this paper, we assume three preferences in the population: (i) never wearing a mask; (ii) wearing a mask if and only if infected; (iii) always wearing a mask. We study the epidemic transmission in an open system within the Susceptible-Infected-Recovered (SIR) model framework. We use agent-based Monte Carlo simulation and mean-field differential equations to investigate the model, respectively. Ternary heat maps show that wearing masks is always beneficial in curbing the spread of the epidemic. Based on the model, we investigate the potential implications of different mask designs from the perspective of collective dynamics. The results show that strengthening the filterability of the mask from the face to the outside is more effective in most parameter spaces, because it acts on individuals with both preferences (ii) and (iii). However, when the fraction of individuals always wearing a mask achieves a critical point, strengthening the filterability from outside to the face becomes more effective, because of the emerging hidden reality that the infected individuals become too few to utilize the filterability from their face to outside fully.
		
	\end{abstract}
	
	\begin{multicols}{2}
		\section{Introduction}\label{intro}
		As the COVID-19 pandemic is ravaging the world, the protection of masks is a topic of interest. A news article published in Nature indicates that wearing a surgical mask leads to an $11\%$ drop in risk, while a $5\%$ drop for cloth \cite{peeples2021face}. The protective effect of masks on individuals may seem minimal, but it is also necessary to focus on the protective effect on collectives.
		
		Since Kermack and McKendrick \cite{kermack1927contribution} proposed the Susceptible-Infected-Recovered (SIR) compartment model, various epidemic models have been developed considerably. In the classic SIR model, the population is divided into three compartments: (i) the susceptible ($S$); (ii) the infected ($I$); (ii) the recovered ($R$). Through human-to-human contact or self-healing, individuals flow from one compartment to another. A simple modified version is the SEIR model, which adds an exposed ($E$) compartment to the SIR model. Recently, Barlow {\it et al.} \cite{barlow2020accurate,weinstein2020analytic} derived the analytical solutions of the SIR \cite{barlow2020accurate} and the SEIR \cite{weinstein2020analytic} models. 
		
		Researchers, in recent years, have explored additional factors and mechanisms to the classic epidemic models, such as isolation \cite{wang2020epidemic} and vaccination \cite{wang2019optimal,fu2011imitation,wang2020vaccination,alam2019three,kuga2018more,alam2019game}. The dynamics of the epidemic transmission can also be applied to the information spreading, creating rumor spreading models \cite{zhao2011rumor,zhao2013rumor} or the public opinion dynamics model \cite{wang2020dynamics,wang2021injurious}. From the perspective of verification and validation, the global stability of this class of nonlinear dynamical systems is widely studied \cite{vargas2011global,li2012algebraic,side2016global,guo2006global,sun2010global,muroya2013global}. In particular, researchers have proved the global stability of endemic equilibria in various epidemic models in multigroup populations \cite{guo2006global,sun2010global,muroya2013global}, which are general cases of the model proposed in this work. A common approach to prove global stability is constructing a Lyapunov function (not limited to epidemiology, but also widely applied to other complex systems such as evolutionary dynamics \cite{cheng2021behavioral,cheng20222pns}), which measures the system's ``energy.'' If the energy continues to decay, then the system will stabilize at an equilibrium point. 
		
		When it comes to the protective effect of masks, several works \cite{li2020effect,gondim2021preventing,auger2021threshold,lasisi2021modeling,han2021effects} are noticed to have emerged in the COVID-19 period after 2020. Li {\it et al.} \cite{li2020effect} treated whether people wear masks or not as an evolutionary game. Gondim \cite{gondim2021preventing} considered masks in the SEIR model and validated the model by real-world data. Auger and Moussaoui \cite{auger2021threshold} studied the confinement’s release threshold, taking the masks into account. Lasisi and Adeyemo \cite{lasisi2021modeling} modeled the effect of wearing masks on COVID-19 infection dynamics. Han {\it et al.} \cite{han2021effects} investigated the effect of three different preferences on wearing a mask. 
		
		Based on the existing literature, we find the previous works on masks have three shortcomings. First, when classifying the population into three categories with different preferences on wearing masks according to their assumptions, there is no work classifying them into three independent variables. They set only two variables as the fraction of two categories, and the remaining category’s fraction is one minus these two variables. This leads to an inability to ensure constant relative proportions of the other two categories when investigating the effect of the proportion of a certain category. Second, previous work did not carry out a complete analysis of the stability of their models. This makes verification and validation challenging. Third, only focusing on the effect of masks on epidemic spreading, there is no previous work considering providing applications of the epidemic models to the design of masks itself.
		
		This paper builds a general epidemic model in an open system considering three different preferences on wearing masks. We start from a set of agent-based rules, and use mean-field analysis to verify and validate the model. In addition to filling in the gaps of previous work by treating three preferences as independent variables and considering global stability analysis, we explore \cite{page2018model,serge2016sociophysics} the effect of different preferences towards wearing masks on the epidemic transmission through our model's eyes. Considering that in the traditional perception, masks are designed at an individual level, we also try to reveal the design strategies of the masks by the collective dynamics based on our model.
		
		\section{Model}\label{model}
		There is an epidemic disease spreading in the system. To prevent this epidemic, individuals hold different preferences for wearing masks. Concerning the infection state, we divide the population into: (i) the susceptible ($x$); (ii) the infected ($y$); (iii) the recovered ($z$). In terms of different preferences towards wearing masks, we divide the population into: (i) those who never wear masks (subscript 0); (ii) those who wear masks if and only if infected (subscript 1); (iii) those who always wear masks (subscript 2). Therefore, we have up to 9 categories according to different combinations of the classification of the two dimensions mentioned above. 
		
		Before describing evolutionary rules, we list the definition of our mathematical symbols in Table~\ref{tab1}.
		
		\begin{table*}[!htbp]
			\caption{\label{tab1}The definition of mathematical symbols}
			\centering
			\begin{tabular}{ll}
				\hline
				Symbol & Definition                                                                   \\ \hline
				$x_0$ & The number of susceptible individuals never wearing a mask.                   \\
				$y_0$ & The number of infected individuals never wearing a mask.                      \\
				$z_0$ & The number of recovered individuals never wearing a mask.                     \\
				$x_1$ & The number of susceptible individuals wearing a mask if and only if infected. \\
				$y_1$ & The number of infected individuals wearing a mask if and only if infected.    \\
				$z_1$ & The number of recovered individuals wearing a mask if and only if infected.   \\
				$x_2$ & The number of susceptible individuals always wearing a mask.                  \\
				$y_2$ & The number of infected individuals always wearing a mask.                     \\
				$z_2$ & The number of recovered individuals always wearing a mask.                    \\
				$n$ & The number of individuals in the system.                                      \\
				$\Lambda$ & The number of new individuals entering the system within unit time.           \\
				$\mu$ & The rate of natural death.                                                    \\
				$r$ & The rate of recovering.                                                       \\
				$\alpha$ & The rate of human-to-human infection.                                         \\
				$\varepsilon_0$ & The fraction of new individuals never wearing a mask.                         \\
				$\varepsilon_1$ & The fraction of new individuals wearing a mask if and only if infected.       \\
				$\varepsilon_2$ & The fraction of new individuals always wearing a mask.                        \\
				$p_I$ & The protective effect produced when an infected individual wears a mask.      \\
				$p_S$ & The protective effect produced when a susceptible individual wears a mask.    \\ \hline
			\end{tabular}
		\end{table*}
		
		\subsection{The agent-based rules}\label{abm}
		
		Consider an open system containing initially $n\big|_{t=0}$ agents (i.e., individuals). Within a Monte Carlo step, an agent $i$ is randomly selected, and the following parallel events occur.
		
		(1) If agent $i$ is susceptible, we again select an agent $j$ randomly. If agent $j$ is infected, then agent $i$ is infected with a probability $\alpha$ ($\alpha>0$). If agent $j$ wears a mask, then agent $i$ spares from infection with a probability $p_I$ ($0<p_I<1$). If agent $i$ wears a mask, then agent $i$ spares from infection with a probability $p_S$ ($0<p_S<1$). 
		
		(2) If agent $i$ is infected, then it recovers with a probability $r$ (the average infection cycle is $1/r$). This does not happen at the same Monte Carlo step as the event~(1). 
		
		(3) Agent $i$ naturally dies with a probability $\mu$ (the average lifespan is $1/\mu$). We do not consider deaths due to the epidemic. 
		
		To ensure the population remains almost unchanged, we must let new agents enter the system. We set the following very first event in a Monte Carlo step, where the number (0) means it happens before the event~(1).
		
		(0) A new agent enters the system with a probability $p_e$. The agent's personal preference determines it never wears a mask with a probability $\varepsilon_0$ ($0<\varepsilon_0<1$), wears a mask if and only if infected with a probability $\varepsilon_1$ ($0<\varepsilon_1<1$), or always wears a mask with a probability $\varepsilon_2$ ($0<\varepsilon_2<1$), yielding $\varepsilon_0+\varepsilon_1+\varepsilon_2=1$. The preference of an agent on masks does not change over time. 
		
		For the population $n$ remaining almost unchanged with time, we let one time step $t$ contains $n\big|_{t=0}$ Monte Carlo steps, such that each agent can be selected once on average. Therefore, the expected number of new agents entering the system within one time step $t$ is $\Lambda=n^* p_e$. The solution is $p_e=\mu$ (see Theorem~\ref{thm1}). 
		
		\subsection{The mean-field equations}\label{meanf}
		
		Performing mean-field analysis, we can approximate the agent-based dynamics into a set of differential equations. We do not dwell on the principles of mean-field analysis, but only explain some important points. (i) Within unit time, each agent is selected once on average, such that the number of events descending in a category is the population in the category. (ii) The probability of selecting an infected agent never wearing a mask from the population is $y_0/n$, and $y_1/n$, $y_2/n$ for selecting an infected agent holding the other two preferences, respectively. (iii) Thanks to the mask, the probability of sparing from infection is $p_I$ or $p_S$, which means the probability of infection is $(1-p_I)$ or $(1-p_S)$. (iv) If there are two layers of protection, they must be both breached for the infection to succeed. 
		
		We denote the state of the system by a vector $\mathbf{\Psi}$, containing the population in nine categories. The mean-field differential equations depicting the agent-based dynamics is 
		\begin{equation}\label{system}
			\dot{\mathbf{\Psi}}=
			\begin{pmatrix}
				\dot{x}_0 \\
				\dot{y}_0 \\
				\dot{z}_0 \\
				\dot{x}_1 \\
				\dot{y}_1 \\
				\dot{z}_1 \\
				\dot{x}_2 \\
				\dot{y}_2 \\
				\dot{z}_2
			\end{pmatrix},
		\end{equation}
		where
		\begin{equation*}
			\left\{\begin{aligned}
				\dot{x}_0=&~\varepsilon_0 \Lambda-\alpha x_0 [y_0+(1-p_I)(y_1+y_2)]/n-\mu x_0, \\
				\dot{y}_0=&~\alpha x_0 [y_0+(1-p_I)(y_1+y_2)]/n-ry_0-\mu y_0, \\
				\dot{z}_0=&~ry_0-\mu z_0, \\
				\dot{x}_1=&~\varepsilon_1 \Lambda-\alpha x_1 [y_0+(1-p_I)(y_1+y_2)]/n-\mu x_1, \\
				\dot{y}_1=&~\alpha x_1 [y_0+(1-p_I)(y_1+y_2)]/n-ry_1-\mu y_1, \\
				\dot{z}_1=&~ry_1-\mu z_1, \\
				\dot{x}_2=&~\varepsilon_2 \Lambda-\alpha (1-p_S)x_2 [y_0+(1-p_I)(y_1+y_2)]/n \\
				&-\mu x_2, \\
				\dot{y}_2=&~\alpha (1-p_S)x_2 [y_0+(1-p_I)(y_1+y_2)]/n-ry_2 \\
				&-\mu y_2, \\
				\dot{z}_2=&~ry_2-\mu z_2.
			\end{aligned}\right.
		\end{equation*}
		
		The system depicted by Eq.~(\ref{system}) is an extended version of the SIR model with a standard incidence rate.
		
		\section{Results and discussion}\label{result}
		In this section, we demonstrate the numerical results of the model from both Monte Carlo simulation and mean-field equations. The algorithm of the Monte Carlo simulation was described in Section~\ref{abm}. In the numerical simulation of mean-field equations, we use iteration $\mathbf{\Psi}(t+\Delta t)=\mathbf{\Psi}(t)+\dot{\mathbf{\Psi}}(t)\Delta t$, where $\Delta t=0.01$. 
		
		We set three statistical measures: (i) the proportion of susceptible, $p_x=(x_0+x_1+x_2)/n$; (ii) the proportion of infected, $p_y=(y_0+y_1+y_2)/n$; (iii) the proportion of recovered, $p_z=(z_0+z_1+z_2)/n$. 
		
		\subsection{Impact of masks on epidemic spreading}\label{effectmask}
		Figures~\ref{time1} and \ref{time2} show the time evolution of $p_x$, $p_y$, and $p_z$ with different parameters and initial conditions.
		
		\begin{figure*}[!htbp]
			\centering 
			\includegraphics[width=12cm]{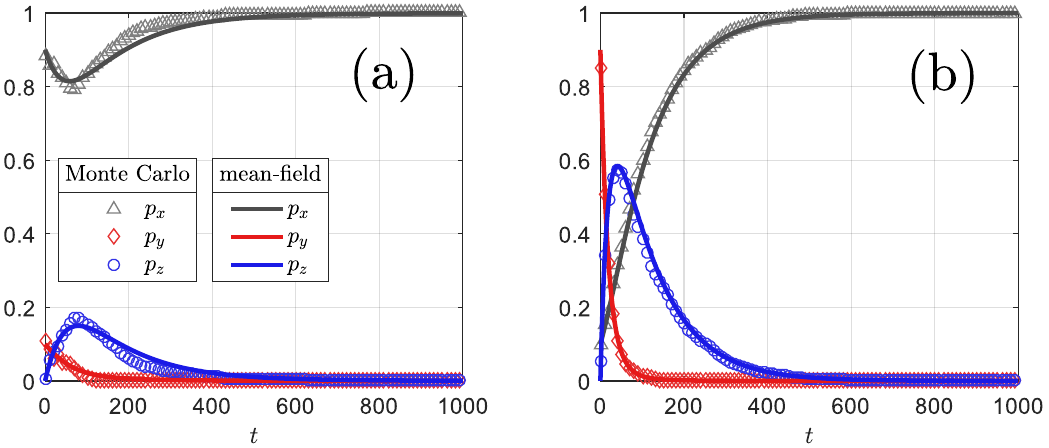}
			\caption{\label{time1}
				Time evolution of $p_x$, $p_y$, and $p_z$. The results of Monte Carlo simulation and mean-field equations are presented together. $\alpha=0.2$, $\mu=0.01$, $\Lambda=10$, $r=0.05$, $\varepsilon_0=0.1$, $\varepsilon_1=0.1$, $\varepsilon_2=0.8$, $p_I=0.75$, $p_S=0.25$. (a) $p_x\big|_{t=0}=0.9$, $p_y\big|_{t=0}=0.1$, $p_z\big|_{t=0}=0$. (b) $p_x\big|_{t=0}=0.1$, $p_y\big|_{t=0}=0.9$, $p_z\big|_{t=0}=0$. (a)(b) $x_0\big|_{t=0}=x_1\big|_{t=0}=x_2\big|_{t=0}$ in $x\big|_{t=0}$, as well as $y\big|_{t=0}$ and $z\big|_{t=0}$.}
		\end{figure*}
		
		\begin{figure*}[!htbp]
			\centering 
			\includegraphics[width=12cm]{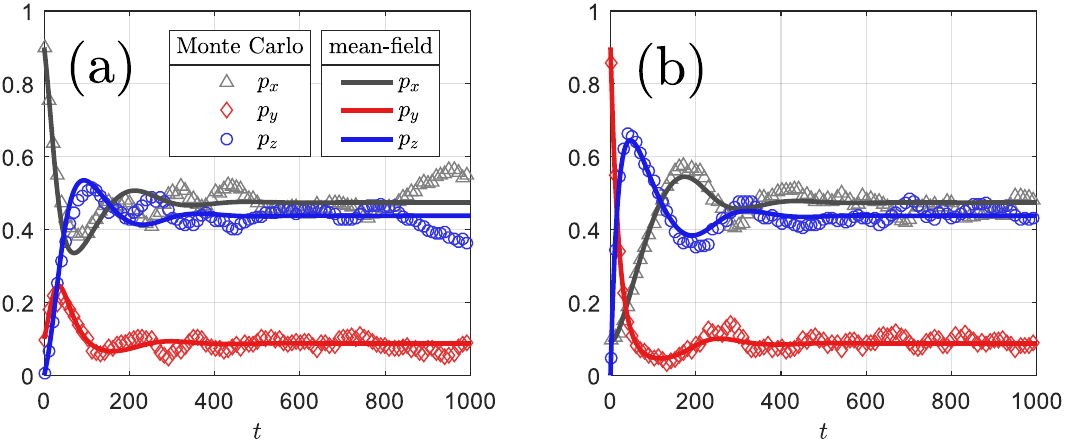}
			\caption{\label{time2}
				Time evolution of $p_x$, $p_y$, and $p_z$. The results of Monte Carlo simulation and mean-field equations are presented together. $\alpha=0.2$, $\mu=0.01$, $\Lambda=10$, $r=0.05$, $\varepsilon_0=0.3$, $\varepsilon_1=0.1$, $\varepsilon_2=0.6$, $p_I=0.5$, $p_S=0.05$. (a) $p_x\big|_{t=0}=0.9$, $p_y\big|_{t=0}=0.1$, $p_z\big|_{t=0}=0$. (b) $p_x\big|_{t=0}=0.1$, $p_y\big|_{t=0}=0.9$, $p_z\big|_{t=0}=0$. (a)(b) $x_0\big|_{t=0}=x_1\big|_{t=0}=x_2\big|_{t=0}$ in $x\big|_{t=0}$, as well as $y\big|_{t=0}$ and $z\big|_{t=0}$.}
		\end{figure*}
		
		From Figures~\ref{time1} and \ref{time2}, we find the proportions of different individuals always achieve stability after time evolution. The results of the Monte Carlo simulation fluctuate, while the results of mean-field equations are stable. They corroborate each other. In addition, we obverse two phenomena. First, the epidemic may either die out or exist at the end, dependent on different parameters. Second, with the same parameters and different initial conditions, the steady-states are the same. 
		
		Next, in the heat maps Figure~\ref{sta1} and \ref{sta2}, we present the steady-states of $p_y$ and $p_x$ as a ternary function of $\varepsilon_0$, $\varepsilon_1$, $\varepsilon_2$, respectively. The results of the Monte Carlo simulation are the average of the last 200 time steps ($t$). The results of mean-field equations are retrieved from the state $\mathbf{\Psi}(t+\Delta t)$ when $\max\{|\mathbf{\Psi}(t)|\Delta t\}<0.0001$.
		
		\begin{figure*}[!htbp]
			\centering 
			\includegraphics[width=14.5cm]{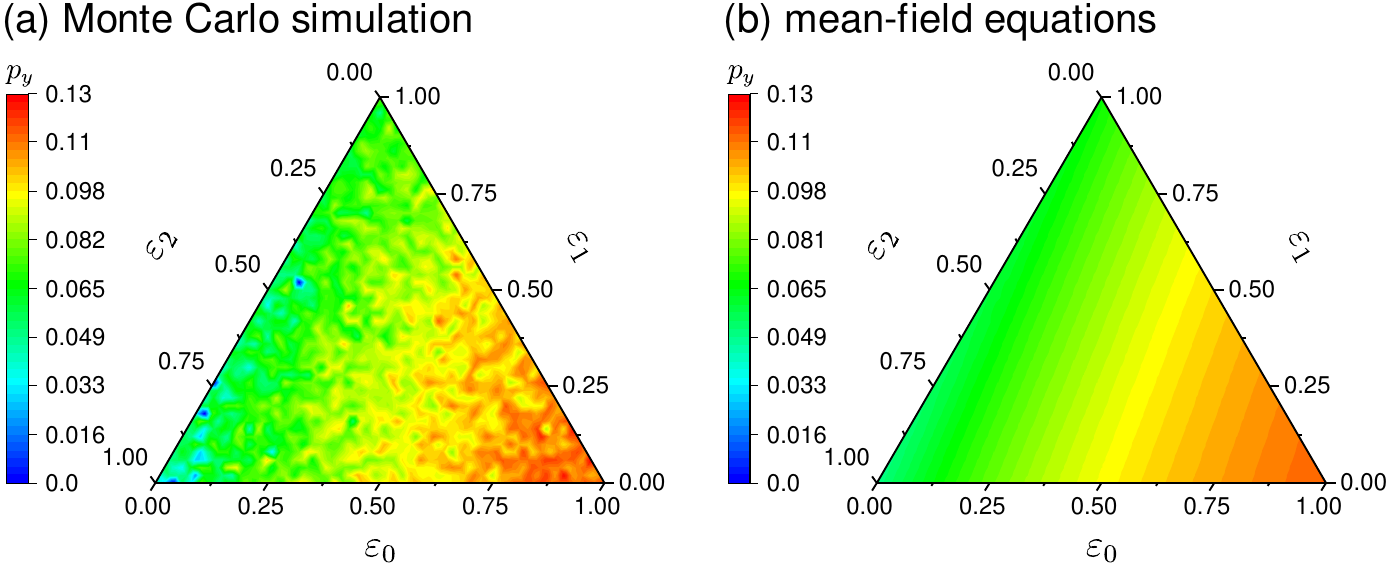}
			\caption{\label{sta1}
				The steady-state of $p_y$ as a ternary function of $\varepsilon_0$, $\varepsilon_1$, $\varepsilon_2$. (a) Monte Carlo simulation. (b) Mean-field equations. $\alpha=0.2$, $\mu=0.01$, $\Lambda=10$, $r=0.05$, $p_I=0.5$, $p_S=0.05$.}
		\end{figure*}
		
		\begin{figure*}[!htbp]
			\centering 
			\includegraphics[width=14.5cm]{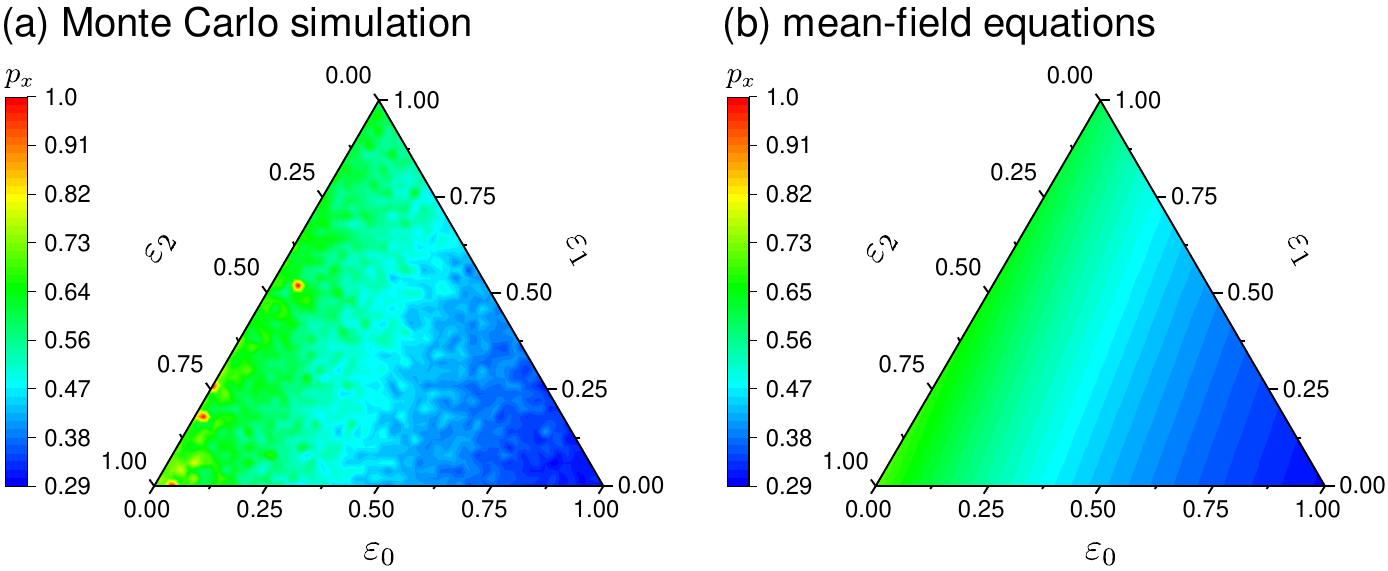}
			\caption{\label{sta2}
				The steady-state of $p_x$ as a ternary function of $\varepsilon_0$, $\varepsilon_1$, $\varepsilon_2$. (a) Monte Carlo simulation. (b) Mean-field equations. $\alpha=0.2$, $\mu=0.01$, $\Lambda=10$, $r=0.05$, $p_I=0.5$, $p_S=0.05$.}
		\end{figure*}
		
		From Figures~\ref{sta1} and \ref{sta2}, we find that the results of Monte Carlo simulation and mean-field equations corroborate each other. In Figure~\ref{sta1}, we observe that more individuals wearing masks reduce the proportion of infected individuals in the population. In particular, always wearing a mask has a better effect on reducing infected individuals. In Figure~\ref{sta2}, we observe that more individuals wearing masks increases the proportion of susceptible individuals in the population. Different from increasing recovered individuals, it means that more people are spared from getting infected once. Also, always wearing a mask has a better effect on increasing susceptible individuals (Figure~\ref{sta2}), making more individuals spared from being infected. 
		
		\subsection{Potential implication of different mask design}\label{applymask}
		Based on our model, we can reveal the potential implication of different mask designs. When an infected individual wears a mask, it is the filterability from the face to the outside that provides protection (to the population), and when a susceptible individual wears a mask, the filterability from the outside to the face matters. 
		
		We show in Figure~\ref{maskfig}~(a) and (b) the steady infected fraction $p_y$ as a binary function of the protective effect produced when an infected individual wears a mask ($p_I$) and when a susceptible individual wears a mask ($p_S$). 
		
		\begin{figure*}[!htbp]
			\centering 
			\includegraphics[width=\textwidth]{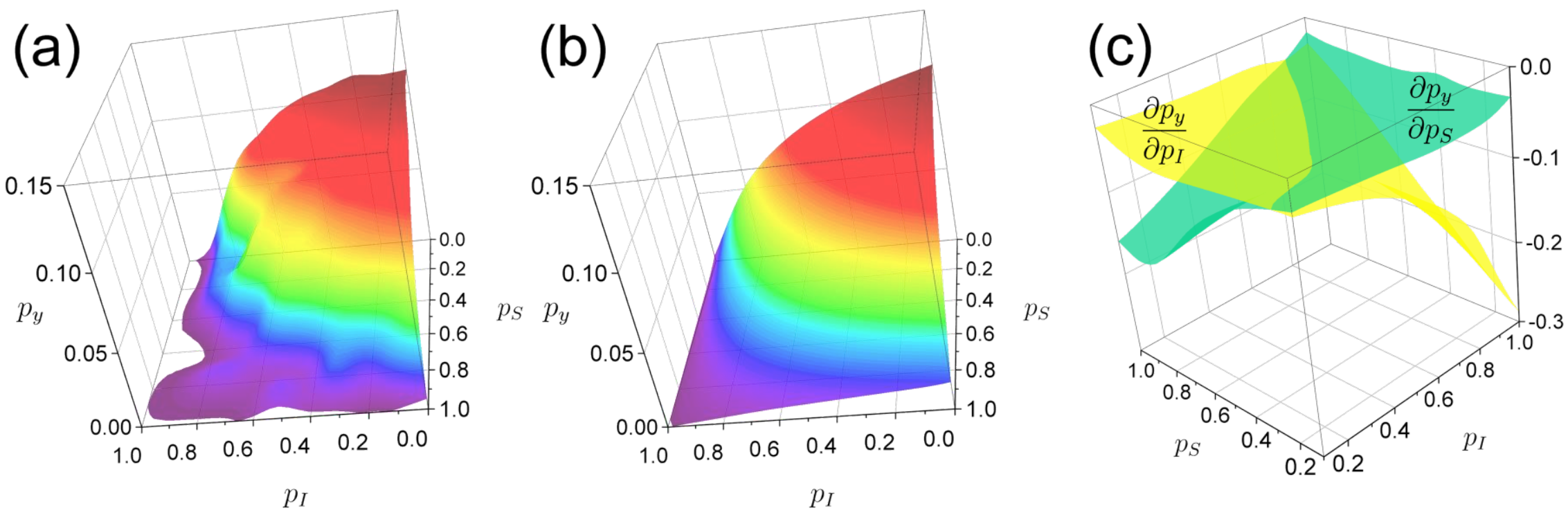}
			\caption{\label{maskfig}
				(a) Monte Carlo simulation. The steady-state of $p_y$ as a binary function of $p_I$, $p_S$. (b) Mean-field equations. The steady-state of $p_y$ as a binary function of $p_I$, $p_S$. (c) Mean-field equations. The steady-state of $\partial p_y/\partial p_I$ and $\partial p_y/\partial p_S$ as a binary function of $p_I$, $p_S$. $\alpha=0.2$, $\mu=0.01$, $\Lambda=10$, $r=0.05$, $\varepsilon_0=0.3$, $\varepsilon_1=0.1$, $\varepsilon_2=0.6$, $p_I=0.5$, $p_S=0.05$.}
		\end{figure*}
		
		We can observe that in both the Monte Carlo simulation [Figure~\ref{maskfig}(a)] and mean-field equations [Figure~\ref{maskfig}(b)], an increase in protective effect $p_I$ or $p_S$ leads to a decrease in the steady infected fraction $p_y$. On this basis, we further ask which one in increasing $p_I$ or $p_S$ is more effective in reducing the infected fraction? 
		
		We show in Figure~\ref{maskfig}(c) the steady-state of $\partial p_y/\partial p_I$ and $\partial p_y/\partial p_S$ as a binary function of $p_I$ and $p_S$. If $\partial p_y/\partial p_I<\partial p_y/\partial p_S$, increasing the unit protective effect from the infected side is more conducive to reducing the infected fraction, and vice versa. Intuitively, increasing $p_I$ should have always been more conducive than increasing $p_S$, because the former acts on individuals with two preferences, (ii) those who wear masks if and only if infected and (iii) those who always wear masks. In contrast, the latter only acts on individuals with one preference, (iii) those who always wear masks. Increasing $p_I$ obviously has a broader scope of action than the increasing $p_S$ and covers the latter’s population. However, Figure~\ref{maskfig}(c) presents a different phenomenon. This indicates that we can provide the designs of the masks with different insights from the group dynamics. We will give this further analysis in Section~\ref{masksec}.
		
		\section{Verification and validation}\label{veri}
		This section verifies and validates the properties that we concluded in numerical results by analyzing them at a mathematical level.
		
		\subsection{The total population dynamics}\label{pop}
		The total population dynamics follows Theorem~\ref{thm1}.
		
		\begin{thm}\label{thm1}
			For $t\to \infty$, we have $n\to \Lambda/\mu$. 
		\end{thm} 
		\begin{proof}
			According to Eq.~(\ref{system}),
			\begin{align}\label{popdynamics}
				\dot{n}&=\dot{x_0}+\dot{y_0}+\dot{z_0}+\dot{x_1}+\dot{y_1}+\dot{z_1}+\dot{x_2}+\dot{y_2}+\dot{z_2}\nonumber \\
				&=\Lambda-\mu n.
			\end{align}
			Solving Eq.~(\ref{popdynamics}), we get
			\begin{equation}\label{ntime}
				n=\left(
				n\big|_{t=0}-\frac{\Lambda}{\mu}
				\right)
				\mathrm{e}^{-\mu t}+\frac{\Lambda}{\mu}.
			\end{equation}
			From Eq.~(\ref{ntime}), we complete the proof that $n\to \Lambda/\mu$ for $t\to \infty$. We give this significant value a symbol $n^*$,
			\begin{equation}\label{nvalue}
				n^*=\lim_{t\to \infty}n=\frac{\Lambda}{\mu}.
			\end{equation}
			In the same way, we can also prove that $x_0+y_0+z_0\to \varepsilon_0 \Lambda/\mu$, $x_1+y_1+z_1\to \varepsilon_1 \Lambda/\mu$, $x_2+y_2+z_2\to \varepsilon_2 \Lambda/\mu$ for $t\to \infty$.
		\end{proof}
		
		Theorem~\ref{thm1} gives us another important insight: when discussing the steady-state, we can substitute $n$ for $n^*=\Lambda/\mu$ in the system of Eq.~(\ref{system}).
		
		\subsection{The basic reproduction number}\label{r0sec}
		The basic reproduction number $\mathcal{R}_0$ is one of the most important measures in epidemiology. It can assist in analyzing both the stability of the system and the effect of parameters on the epidemic transmission. 
		
		We let $\dot{\mathbf{\Psi}}=\mathbf{0}$ and solve for the epidemic-free equilibrium, denoted by $\mathbf{\Psi}^*$, 
		\begin{equation}\label{epifree}
			\mathbf{\Psi}^*=\frac{\Lambda}{\mu}\left(\varepsilon_0,0,0,\varepsilon_1,0,0,\varepsilon_2,0,0\right)^\mathrm{T}.
		\end{equation}
		
		Then, we can follow the method in Ref.~\cite{van2002reproduction} to find the basic reproduction number (see Appendix~\ref{reprodu}):
		\begin{equation}\label{r0}
			\mathcal{R}_0=\frac{\alpha}{r+\mu}[\varepsilon_0+(1-p_I) \varepsilon_1+(1-p_S)(1-p_I) \varepsilon_2].
		\end{equation}
		
		The basic reproduction number reveals the following theorem.
		
		\begin{thm}\label{thm2}
			The epidemic-free equilibrium $\mathbf{\Psi}^*$ is locally asymptotically stable if $\mathcal{R}_0<1$, and the epidemic-free equilibrium $\mathbf{\Psi}^*$ is not stable if $\mathcal{R}_0>1$.
			
			(See Ref.~\cite{van2002reproduction} for proof)
		\end{thm} 
		
		Substituting the parameters in Figure~\ref{time1} into Eq.~(\ref{r0}), we can calculate $\mathcal{R}_0=0.9167<1$, which means the epidemic-free equilibrium is stable, consistent with that shown in Figure~\ref{time1}. Similarly, substituting parameters in Figure~\ref{time2} produces  $\mathcal{R}_0=2.1167>1$, such that the epidemic-free equilibrium is not stable, which is also consistent with that shown in Figure~\ref{time2}.
		
		\subsection{Epidemic-free and endemic equilibria}\label{globalsec}
		We separate $\mathbf{\Psi}$ into 
		\begin{equation}\label{sep}
			\left\{\begin{aligned}
				\mathbf{\Phi}_1&=\left(x_0,y_0,x_1,y_1,x_2,y_2\right)^\mathrm{T},\\
				\mathbf{\Phi}_2&=\left(z_0,z_1,z_2\right)^\mathrm{T}.
			\end{aligned}\right.
		\end{equation}
		
		From Eq.~(\ref{system}), we can assert that the steady-state of $\mathbf{\Phi}_1$ can determine the steady-state of $\mathbf{\Phi}_2$, and $\mathbf{\Phi}_2$ does not affect the evolution of $\mathbf{\Phi}_1$. Therefore, the stability of $\mathbf{\Psi}$ is equivalent to (i) the stability of $\mathbf{\Phi}_1$ and (ii) the stability of $\mathbf{\Phi}_2$ when $\mathbf{\Phi}_1$ achieves stability. 
		
		We first prove the global stability of the epidemic-free equilibrium by constructing the Lyapunov function.
		
		\begin{thm}\label{thm3}
			The epidemic-free equilibrium $\mathbf{\Psi}^*$ is global asymptotically stable in $\mathbb{R}_{\geq 0}^9$ if $\mathcal{R}_0<1$.
		\end{thm}
		\begin{proof}
			Consider the Lyapunov function $\mathcal{L}(\mathbf{\Phi}_1)$ in $\mathbb{R}_{\geq 0}^6$,
			\begin{align}\label{L1}
				\mathcal{L}(\mathbf{\Phi}_1)=
				&~\frac{(x_0-x_0^*)^2}{2x_0^*}+y_0\nonumber \\
				&+(1-p_I)\left[\frac{(x_1-x_1^*)^2}{2x_1^*}+y_1\right]\nonumber \\
				&+(1-p_I)\left[\frac{(x_2-x_2^*)^2}{2x_2^*}+y_2\right].
			\end{align}
			We can conclude that, (i) $\mathcal{L}(\mathbf{\Phi}_1)=0$ when $\mathbf{\Phi}_1=\mathbf{\Phi}_1^*$, (ii) $\mathcal{L}(\mathbf{\Phi}_1)>0$ when $\mathbf{\Phi}_1\neq \mathbf{\Phi}_1^*$. Therefore, $\mathbf{\Phi}_1$ is positive definite in the neighborhood of $\mathbf{\Phi}_1^*$. Secondly, we have (see Appendix~\ref{L1condi}),
			\begin{equation}\label{L1'}
				\dot{\mathcal{L}}(\mathbf{\Phi}_1)\leq 0
			\end{equation}
			when $\mathbf{\Phi}_1\neq \mathbf{\Phi}_1^*$. Note that $\dot{\mathcal{L}}(\mathbf{\Phi}_1)=0$ can be confirmed by Eq.~(\ref{L1calcu}) when $\mathbf{\Phi}_1=\mathbf{\Phi}_1^*$. Therefore, $\mathbf{\Phi}_1$ is negative semi-definite in the neighborhood of $\mathbf{\Phi}_1^*$. Hence, according to Lasalle’s Invariance Principle \cite{la1976stability}, $\mathbf{\Phi}_1^*$ is globally asymptotically stable in $\mathbb{R}_{\geq 0}^6$. Given $\mathbf{\Phi}_1^*$ stable, it is easy to prove the global asymptotic stability of $\mathbf{\Phi}_2^*$ in $\mathbb{R}_{\geq 0}^3$ by constructing Lyapunov function $\mathcal{L}_z(\mathbf{\Phi}_2)=z_0+z_1+z_2$. Therefore, $\mathbf{\Psi}^*$ is global asymptotically stable in $\mathbb{R}_{\geq 0}^9$. 
		\end{proof}
		
		The equation $\dot{\mathbf{\Psi}}=\mathbf{0}$ has two solutions. We denote the second solution by $\mathbf{\Psi}^{**}$. The infected population is non-zero; thus, we call it the endemic equilibrium. This equilibrium $\mathbf{\Psi}^{**}$ corresponds to the results shown in Figure~\ref{time2}. It is not easy to express it analytically, but we can still show its existence condition.
		
		\begin{thm}\label{thm4}
			The endemic equilibrium $\mathbf{\Psi}^{**}$ exists and is unique in $\mathbb{R}_{\geq 0}^9$ if $\mathcal{R}_0>1$.
		\end{thm}
		\begin{proof}
			First, we show the relationship between the existence and uniqueness of positive $y_i^{**}$, $i=0,1,2$ ($y_0^{**}>0$, $y_1^{**}>0$, $y_2^{**}>0$) and $\mathcal{R}_0>1$ (see Appendix~\ref{endeexist}). Then, the existence and uniqueness of $x_i^{**}$, $z_i^{**}$, $i=0,1,2$ can be naturally confirmed, hence the existence and uniqueness of $\mathbf{\Psi}^{**}$.
		\end{proof}
		
		Theorem~\ref{thm3} validates that in Figure~\ref{time1}, the steady-state with the same parameters is independent of the initial conditions. Theorem~\ref{thm4} is consistent with Figure~\ref{time2}.
		
		\subsection{Robustness analysis for the effect of wearing masks}\label{robustsec}
		The basic reproduction number measures the average number of individuals that an infected individual can transmit the epidemic. The higher the basic reproduction number, the more severe the epidemic.
		
		Analyzing Eq.~(\ref{r0}), we can see that the coefficient before $\varepsilon_0$ is 1, and $(1-p_I)$ for $\varepsilon_1$, and $(1-p_S)(1-p_I)$ for $\varepsilon_2$. Since $p_S>0$, $p_I>0$, we have $(1-p_S)(1-p_I)<1-p_I<1$. Considering the constraints: $0\leq \varepsilon_0\leq 1$, $0\leq \varepsilon_1\leq 1$, $0\leq \varepsilon_2\leq 1$, $\varepsilon_0+\varepsilon_1+\varepsilon_2=1$, we know the following facts.
		(i) $\mathcal{R}_0$ takes the minimum when $\varepsilon_0=0$, $\varepsilon_1=0$, $\varepsilon_2=1$. 
		(ii) $\mathcal{R}_0$ takes the maximum when $\varepsilon_0=1$, $\varepsilon_1=0$, $\varepsilon_2=0$. 
		Therefore, everyone always wearing a mask minimizes the epidemic severity, while no one wearing masks maximizes the epidemic severity.

		\subsection{Application to mask design}\label{masksec}
		The basic reproduction number $\mathcal{R}_0$ can play the same role as the infected fraction $p_y$ in measuring the outbreak severity. From Eq.~(\ref{r0}), we see that the coefficient $(1-p_I)$ acts on both $\varepsilon_1$ and $\varepsilon_2$ while $(1-p_S)$ only acts on $\varepsilon_2$, which means the mask protection from the infected side acts on two population categories and that from the susceptible side acts on only one. This brings us a misleading intuition that increasing the protective effect from the infected side is always more conducive. However, we can reveal a hidden different reality by group dynamics.
		
		We write the partial derivatives of $\mathcal{R}_0$ with respect to $p_I$ and $p_S$ in Eq.~(\ref{r0topi}) and (\ref{r0tops}).
		\begin{equation}\label{r0topi}
			\frac{\partial \mathcal{R}_0}{\partial p_I}=-\frac{\alpha}{r+\mu}[\varepsilon_1+(1-p_S)\varepsilon_2],
		\end{equation}
		\begin{equation}\label{r0tops}
			\frac{\partial \mathcal{R}_0}{\partial p_S}=-\frac{\alpha}{r+\mu}(1-p_I)\varepsilon_2.
		\end{equation}
		
		Increasing the protective effect from the infected side is better than that of the susceptible one means
		\begin{equation}\label{infebettercondi}
			\frac{\partial \mathcal{R}_0}{\partial p_I}<\frac{\partial \mathcal{R}_0}{\partial p_S}
		\end{equation}
		or
		\begin{equation}\label{infebetter}
			\frac{\varepsilon_1}{\varepsilon_2}>p_S-p_I,
		\end{equation}
		and, increasing the protective effect from the susceptible side is better than that of the infected one means
		\begin{equation}\label{susbettercondi}
			\frac{\partial \mathcal{R}_0}{\partial p_S}<\frac{\partial \mathcal{R}_0}{\partial p_I}
		\end{equation}
		or
		\begin{equation}\label{susbetter}
			\frac{\varepsilon_1}{\varepsilon_2}<p_S-p_I,
		\end{equation}
		
		We can discuss the parameter space in two cases. First, if $p_S<p_I$, then Eq.~(\ref{infebettercondi}) always holds. Second, if $p_S>p_I$, then Eq.~(\ref{susbettercondi}) does not always hold. This suggests that the ``intuitive" phenomena (increasing  $p_I$ is more effective) occupy more parameter space, which can be verified by Figure~\ref{maskfig}(c).
		
		For the latter case, $p_S>p_I$, we can transform Eq.~(\ref{susbetter}) into $\varepsilon_2>\varepsilon_1/(p_S-p_I)$. This indicates that if the fraction of individuals always wearing a mask ($\varepsilon_2$) exceeds a critical point, $\varepsilon_1/(p_S-p_I)$, then, increasing $p_S$ is more effective than increasing $p_I$, even if $p_I$ can act on both category $\varepsilon_1$ and $\varepsilon_2$. We can interpret this result in daily language by taking into account the fraction of existing infected individuals. As we showed in Section~\ref{robustsec}, $\mathcal{R}_0$ decreases (i.e., infected individuals increasing) with an increase in $\varepsilon_2$. In this way, the critical point of $\varepsilon_2$ is rational to exist, over which the infected individuals are too few to exert the protective effect that the mask produces on their side. 
		
		\section{Conclusion}\label{conclusion}
		Although most masks have little to no effect on personal protection \cite{peeples2021face}, we are still interested in the protective effects of masks on a population. We proposed a general epidemic model in the classic SIR framework considering three different preferences towards wearing masks. Some individuals never wear masks; others wear masks if and only if infected, and some always wear masks. We started from agent-based rules and used a set of mean-field differential equations to approximate the model. The results of the two corroborate each other. In this work, the three preferences are independent of each other. 
		
		The first aspect is the effect of masks on epidemics. The ternary heat maps revealed that wearing masks can reduce the number of infected individuals and increase the number of susceptible individuals. We provided the global stability analysis of the results and showed the robustness of the effectiveness of masks by analyzing the basic reproduction number of the epidemic. We concluded that wearing masks are beneficial to the control of epidemics. 
		
		The second aspect is the application of the epidemic model to mask design. The protective effect from the infected side ($p_I$) can be understood as the filterability of the mask from the face to the outside against viruses, while the protective effect to the susceptible side ($p_S$) can be interpreted as the filterability from the outside to the face. This can be influenced by the material and design of the mask \cite{boraey2021analytical}, and we analyzed which side strengthening would provide better results. We showed that strengthening the infected side is more effective in most parameter spaces. This is intuitive since strengthening the infected side acts on two categories of individuals (those wearing masks only if infected and those always wearing masks), while strengthening the susceptible side acts on only one category (those always wearing masks). However, there is a hidden reality from the perspective of group dynamics. We found that once the fraction of individuals always wearing masks exceeds a critical point, $\varepsilon_2>\varepsilon_1/(p_S-p_I)$, then, strengthening the susceptible side becomes more effective. This is because the preference of always wearing masks reduces the infected fraction in the population, so that the infected individuals are too few to exert the protective effect of masks produced on their side. In the daily language, both the cases above seem to make sense. However, noticing the latter case from the group perspective and further giving the mask design strategies according to parameter spaces are not straightforward without the help of system dynamics.
		
		Real-world situations may have more complexity and different insights. For instance, the underlying assumptions---people's preferences do not change with time, ignores human subjectivity, which has the potential to reveal more insights. In fact, people can change their preference on whether to wear masks by either estimating the epidemic severity (evolutionary games) or being affected by the propaganda of the effectiveness of masks (opinion dynamics). In this way, future work may consider time-dependent preferences, and apply any modified model to mask design. 
		
		\section*{Acknowledgement}\label{ackno}
		Publication of this article was funded in part by the George Mason University Libraries Open Access Publishing Fund.
		
		\section*{Data availability}\label{dataavai}
		Data sharing not applicable to this article as no datasets were generated or analyzed during the current study.
		
		\section*{Conflict of interest statement}\label{confli}
		On behalf of all authors, the corresponding author states that there is no conflict of interest.
	\end{multicols}
	
	\begin{appendices}
		\section{Finding the basic reproduction number}\label{reprodu}
		\setcounter{equation}{0}
		\renewcommand{\theequation}{A.\arabic{equation}}
		Let us decompose the infected compartments in Eq.~(\ref{system}) as $\left(\dot{y}_0,\dot{y}_1,\dot{y}_2\right)^\mathrm{T}=\mathcal{F}-\mathcal{V}$, where
		\begin{equation}
			\mathcal{F}=
			\begin{pmatrix}
				\mathcal{F}_{y_0} \\
				\mathcal{F}_{y_1} \\
				\mathcal{F}_{y_0}
			\end{pmatrix}
			=
			\begin{pmatrix}
				\alpha x_0 [y_0+(1-p_I)(y_1+y_2)]/n^* \\
				\alpha x_1 [y_0+(1-p_I)(y_1+y_2)]/n^* \\
				\alpha (1-p_S)x_2 [y_0+(1-p_I)(y_1+y_2)]/n^*
			\end{pmatrix},
			\label{F}
		\end{equation}
		\begin{equation}
			\mathcal{V}=
			\begin{pmatrix}
				\mathcal{V}_{y_0} \\
				\mathcal{V}_{y_1} \\
				\mathcal{V}_{y_0}
			\end{pmatrix}
			=
			\begin{pmatrix}
				ry_0+\mu y_0 \\
				ry_1+\mu y_1 \\
				ry_2+\mu y_2
			\end{pmatrix}.
			\label{V}
		\end{equation}
		
		Solve for the Jacobian matrix of $\mathcal{F}$ and $\mathcal{V}$ at $\mathbf{\Psi}^*$, denoted by $\mathbf{F}$ and $\mathbf{V}$, 
		\begin{align}
			\mathbf{F}&=
			\begin{pmatrix}
				\displaystyle \frac{\partial\mathcal{F}_{y_0}}{\partial y_0} & 
				\displaystyle \frac{\partial\mathcal{F}_{y_0}}{\partial y_1} & 
				\displaystyle \frac{\partial\mathcal{F}_{y_0}}{\partial y_2} \\[8pt]
				\displaystyle \frac{\partial\mathcal{F}_{y_1}}{\partial y_0} &
				\displaystyle \frac{\partial\mathcal{F}_{y_1}}{\partial y_1} &
				\displaystyle \frac{\partial\mathcal{F}_{y_1}}{\partial y_2} \\[8pt]
				\displaystyle \frac{\partial\mathcal{F}_{y_2}}{\partial y_0} &
				\displaystyle \frac{\partial\mathcal{F}_{y_2}}{\partial y_1} &
				\displaystyle \frac{\partial\mathcal{F}_{y_2}}{\partial y_2}
			\end{pmatrix}
			(\mathbf{\Psi}^*)\nonumber
			=\frac{1}{n^*}
			\begin{pmatrix}
				\displaystyle \alpha x_0 & 
				\displaystyle \alpha (1-p_I)x_0 & 
				\displaystyle \alpha (1-p_I)x_0 \\
				\displaystyle \alpha x_1 & 
				\displaystyle \alpha (1-p_I)x_1 & 
				\displaystyle \alpha (1-p_I)x_1 \\
				\displaystyle \alpha (1-p_S)x_2 & 
				\displaystyle \alpha (1-p_S)(1-p_I)x_2 & 
				\displaystyle \alpha (1-p_S)(1-p_I)x_2
			\end{pmatrix} \\
			&=\alpha
			\begin{pmatrix}
				\displaystyle x_0 & 
				\displaystyle (1-p_I)x_0 & 
				\displaystyle (1-p_I)x_0 \\
				\displaystyle x_1 & 
				\displaystyle (1-p_I)x_1 & 
				\displaystyle (1-p_I)x_1 \\
				\displaystyle (1-p_S)x_2 & 
				\displaystyle (1-p_S)(1-p_I)x_2 & 
				\displaystyle (1-p_S)(1-p_I)x_2
			\end{pmatrix},
			\label{Fmatri}
		\end{align}
		\begin{equation}
			\mathbf{V}=
			\begin{pmatrix}
				\displaystyle \frac{\partial\mathcal{V}_{y_0}}{\partial y_0} & 
				\displaystyle \frac{\partial\mathcal{V}_{y_0}}{\partial y_1} & 
				\displaystyle \frac{\partial\mathcal{V}_{y_0}}{\partial y_2} \\[8pt]
				\displaystyle \frac{\partial\mathcal{V}_{y_1}}{\partial y_0} &
				\displaystyle \frac{\partial\mathcal{V}_{y_1}}{\partial y_1} &
				\displaystyle \frac{\partial\mathcal{V}_{y_1}}{\partial y_2} \\[8pt]
				\displaystyle \frac{\partial\mathcal{V}_{y_2}}{\partial y_0} &
				\displaystyle \frac{\partial\mathcal{V}_{y_2}}{\partial y_1} &
				\displaystyle \frac{\partial\mathcal{V}_{y_2}}{\partial y_2}
			\end{pmatrix}
			(\mathbf{\Psi}^*)
			=(r+\mu)
			\begin{pmatrix}
				\displaystyle 1 & 
				\displaystyle 0 & 
				\displaystyle 0 \\
				\displaystyle 0 & 
				\displaystyle 1 & 
				\displaystyle 0 \\
				\displaystyle 0 & 
				\displaystyle 0 & 
				\displaystyle 1
			\end{pmatrix}.
			\label{Vmatri}
		\end{equation}
		
		Then, the spectral radius (i.e., maximum eigenvalue) of $\mathbf{F}\cdot \mathbf{V}^{-1}$ is the basic reproduction number $\mathcal{R}_0$, 
		\begin{equation}
			\mathcal{R}_0=\frac{\alpha}{r+\mu}[\varepsilon_0+(1-p_I) \varepsilon_1+(1-p_S)(1-p_I) \varepsilon_2].
		\end{equation}
		Please see Ref.~\cite{van2002reproduction} for more information on how to find the basic reproduction number.
		
		\section{Proof of $\dot{\mathcal{L}}(\mathbf{\Phi}_1)\leq 0$ 	when $\mathbf{\Phi}_1\neq \mathbf{\Phi}_1^*$}\label{L1condi}
		\setcounter{equation}{0}
		\renewcommand{\theequation}{B.\arabic{equation}}
		\begin{align}\label{L1calcu}
			\dot{\mathcal{L}}(\mathbf{\Phi}_1)
			=&\left(\frac{x_0}{x_0^*}-1\right)\dot{x}_0+\dot{y}_0
			+(1-p_I)\left[\left(\frac{x_1}{x_1^*}-1\right)\dot{x}_1+\dot{y}_1\right]
			+(1-p_I)\left[\left(\frac{x_2}{x_2^*}-1\right)\dot{x}_2+\dot{y}_2\right] \nonumber \\
			=&\left(\frac{x_0}{x_0^*}-1\right)\left(\varepsilon_0 \Lambda-\frac{\alpha x_0 [y_0+(1-p_I)(y_1+y_2)]}{n^*} -\mu x_0\right)+\frac{\alpha x_0 [y_0+(1-p_I)(y_1+y_2)]}{n^*} \nonumber \\
			&-ry_0-\mu y_0+(1-p_I)\left(\frac{x_1}{x_1^*}-1\right)\left(\varepsilon_1 \Lambda-\frac{\alpha x_1 [y_0+(1-p_I)(y_1+y_2)]}{n^*} -\mu x_1\right)\nonumber \\
			&+(1-p_I)\left(\frac{\alpha x_1 [y_0+(1-p_I)(y_1+y_2)]}{n^*}-ry_1-\mu y_1\right)\nonumber \\
			&+(1-p_I)\left(\frac{x_2}{x_2^*}-1\right)\left(\varepsilon_2 \Lambda-\frac{\alpha (1-p_S)x_2 [y_0+(1-p_I)(y_1+y_2)]}{n^*} -\mu x_2\right)\nonumber \\
			&+(1-p_I)\left(\frac{\alpha (1-p_S)x_2 [y_0+(1-p_I)(y_1+y_2)]}{n^*}-ry_2-\mu y_2\right)\nonumber \\
			=&-\frac{\mu}{x_0^*}(x_0-x_0^*)^2-\frac{\alpha}{x_0^* n^*}[y_0+(1-p_I)(y_1+y_2)](x_0-x_0^*)^2\nonumber \\
			&+(r+\mu)\left(\frac{\alpha x_0^*}{r+\mu}\times \frac{y_0+(1-p_I)(y_1+y_2)}{n^*}-y_0\right)\nonumber \\
			&-\frac{\mu}{x_1^*}(1-p_I)(x_1-x_1^*)^2-\frac{\alpha}{x_1^* n^*}(1-p_I)[y_0+(1-p_I)(y_1+y_2)](x_1-x_1^*)^2\nonumber \\
			&+(r+\mu)(1-p_I)\left(\frac{\alpha x_1^*}{r+\mu}\times \frac{y_0+(1-p_I)(y_1+y_2)}{n^*}-y_1\right)\nonumber \\
			&-\frac{\mu}{x_2^*}(1-p_I)(x_2-x_2^*)^2-\frac{\alpha}{x_2^* n^*}(1-p_S)(1-p_I)[y_0+(1-p_I)(y_1+y_2)](x_2-x_2^*)^2\nonumber \\
			&+(r+\mu)(1-p_I)\left(\frac{\alpha (1-p_S)x_2^*}{r+\mu}\times \frac{y_0+(1-p_I)(y_1+y_2)}{n^*}-y_2\right).
		\end{align}
		In Eq.~(\ref{L1calcu}), we used $x_0^*=\varepsilon_0\Lambda/\mu$, $x_1^*=\varepsilon_1\Lambda/\mu$, $x_2^*=\varepsilon_2\Lambda/\mu$.
		
		We can further deflate Eq.~(\ref{L1calcu}),
		\begin{align}\label{L1calcu2}
			\dot{\mathcal{L}}(\mathbf{\Phi}_1)
			=&-\frac{\mu}{x_0^*}(x_0-x_0^*)^2-\frac{\alpha}{x_0^* n^*}[y_0+(1-p_I)(y_1+y_2)](x_0-x_0^*)^2\nonumber \\
			&-\frac{\mu}{x_1^*}(1-p_I)(x_1-x_1^*)^2-\frac{\alpha}{x_1^* n^*}(1-p_I)[y_0+(1-p_I)(y_1+y_2)](x_1-x_1^*)^2\nonumber \\
			&-\frac{\mu}{x_2^*}(1-p_I)(x_2-x_2^*)^2-\frac{\alpha}{x_2^* n^*}(1-p_S)(1-p_I)[y_0+(1-p_I)(y_1+y_2)](x_2-x_2^*)^2\nonumber \\
			&+(r+\mu)[y_0+(1-p_I)(y_1+y_2)]\left\{\frac{\alpha}{r+\mu}[\varepsilon_0+(1-p_I)\varepsilon_1+(1-p_S)(1-p_I)\varepsilon_2]-1\right\}\nonumber \\
			=&-\frac{\mu}{x_0^*}(x_0-x_0^*)^2-\frac{\alpha}{x_0^* n^*}[y_0+(1-p_I)(y_1+y_2)](x_0-x_0^*)^2\nonumber \\
			&-\frac{\mu}{x_1^*}(1-p_I)(x_1-x_1^*)^2-\frac{\alpha}{x_1^* n^*}(1-p_I)[y_0+(1-p_I)(y_1+y_2)](x_1-x_1^*)^2\nonumber \\
			&-\frac{\mu}{x_2^*}(1-p_I)(x_2-x_2^*)^2-\frac{\alpha}{x_2^* n^*}(1-p_S)(1-p_I)[y_0+(1-p_I)(y_1+y_2)](x_2-x_2^*)^2\nonumber \\
			&+(r+\mu)[y_0+(1-p_I)(y_1+y_2)](\mathcal{R}_0-1)\nonumber \\
			\leq&~0,
		\end{align}
		which completes the proof.
		
		\section{The existence and uniqueness of $\mathbf{\Psi}^{**}$ when $\mathcal{R}_0>1$}\label{endeexist}
		\setcounter{equation}{0}
		\renewcommand{\theequation}{C.\arabic{equation}}
		Using the equations $\dot{y}_0=0$, $\dot{y}_1=0$, $\dot{y}_2=0$ in $\dot{\mathbf{\Psi}}=\mathbf{0}$ to obtain $x_i^{**}$ as a function of $y_i^{**}$, $i=0,1,2$. Then, substituting the results into the equations $\dot{x}_0=0$, $\dot{x}_1=0$, $\dot{x}_2=0$, 
		\begin{equation}\label{c1}
			\left\{\begin{aligned}
				0=&~\varepsilon_0 \Lambda-(r+\mu)y_0^{**}-\frac{\mu n^*(r+\mu)y_0^{**}}{\alpha [y_0^{**}+(1-p_I)(y_1^{**}+y_2^{**})]}, \\
				0=&~\varepsilon_1 \Lambda-(r+\mu)y_1^{**}-\frac{\mu n^*(r+\mu)y_1^{**}}{\alpha [y_0^{**}+(1-p_I)(y_1^{**}+y_2^{**})]}, \\
				0=&~\varepsilon_2 \Lambda-(r+\mu)y_2^{**}-\frac{\mu n^*(r+\mu)y_2^{**}}{\alpha (1-p_S)[y_0^{**}+(1-p_I)(y_1^{**}+y_2^{**})]}. \\
			\end{aligned}\right.
		\end{equation}
		
		In Eq.~(\ref{c1}), we multiply the first equation by $\alpha/(r+\mu)$, the second equation by $\alpha(1-p_I)/(r+\mu)$, and the third equation by $\alpha(1-p_S)(1-p_I)/(r+\mu)$:
		\begin{equation}\label{c2}
			\left\{\begin{aligned}
				0=&~\frac{\alpha}{r+\mu}\varepsilon_0 \Lambda-\alpha y_0^{**}-\frac{\mu n^*y_0^{**}}{y_0^{**}+(1-p_I)(y_1^{**}+y_2^{**})}, \\
				0=&~\frac{\alpha}{r+\mu}(1-p_I)\varepsilon_1 \Lambda-\alpha (1-p_I)y_1^{**}-\frac{\mu n^*(1-p_I)y_1^{**}}{y_0^{**}+(1-p_I)(y_1^{**}+y_2^{**})}, \\
				0=&~\frac{\alpha}{r+\mu}(1-p_S)(1-p_I)\varepsilon_2 \Lambda-\alpha(1-p_S)(1-p_I)y_2^{**}-\frac{\mu n^*(1-p_I)y_2^{**}}{y_0^{**}+(1-p_I)(y_1^{**}+y_2^{**})}. \\
			\end{aligned}\right.
		\end{equation}
		
		Summing up the three equations in Eq.~(\ref{c2}) and using $n^*=\Lambda/\mu$ (see Eq.~(\ref{nvalue})), we have
		\begin{equation} 		
			\mathcal{R}_0-\alpha[y_0^{**}+(1-p_I)y_1^{**}+(1-p_S)(1-p_I)y_2^{**}]-1=0.
		\end{equation}
		Therefore, to ensure $y_0^{**}+(1-p_I)y_1^{**}+(1-p_S)(1-p_I)y_2^{**}>0$, which is a necessary condition for $y_0^{**}>0$, $y_1^{**}>0$, $y_2^{**}>0$, we have $\mathcal{R}_0>1$. However, we have not yet proved that $\mathcal{R}_0>1$ is a sufficient condition for $y_0^{**}>0$, $y_1^{**}>0$, $y_2^{**}>0$.
		
		According to Eq.~(\ref{c2}), we can ensure $y_0^{**}>0$, $y_1^{**}>0$, $y_2^{**}>0$ if we can confirm $y_0^{**}+(1-p_I)(y_1^{**}+y_2^{**})>0$. We will try to illustrate the opposite case, $y_0^{**}+(1-p_I)(y_1^{**}+y_2^{**})<0$, cannot happen. Let us further write Eq.~(\ref{c2}) as
		\begin{equation}\label{c4}
			\left\{\begin{aligned}
				y_0^{**}=&~\dfrac{\dfrac{\alpha}{r+\mu}\varepsilon_0 \Lambda}{\alpha+\dfrac{\mu n^*}{y_0^{**}+(1-p_I)(y_1^{**}+y_2^{**})}}, \\
				y_1^{**}=&~\dfrac{\dfrac{\alpha}{r+\mu}\varepsilon_1 \Lambda}{\alpha+\dfrac{\mu n^*}{y_0^{**}+(1-p_I)(y_1^{**}+y_2^{**})}}, \\
				y_2^{**}=&~\dfrac{\dfrac{\alpha}{r+\mu}(1-p_S)\varepsilon_2 \Lambda}{\alpha(1-p_S)+\dfrac{\mu n^*}{y_0^{**}+(1-p_I)(y_1^{**}+y_2^{**})}}. \\
			\end{aligned}\right.
		\end{equation}
		Then, it can be judged that $y_0^{**}$ and $y_1^{**}$ have the same sign, because the denominators are equal and the numerators are both positive. The case $y_0^{**}<0$, $y_1^{**}<0$ is possible only if their denominators $\alpha +\mu n^*/[y_0^{**}+(1-p_I)(y_1^{**}+y_2^{**})]<0$. In this case, we have $\alpha(1-p_S) +\mu n^*/[y_0^{**}+(1-p_I)(y_1^{**}+y_2^{**})]<\alpha +\mu n^*/[y_0^{**}+(1-p_I)(y_1^{**}+y_2^{**})]<0$, which means $y_2^{**}<0$ as well. Then, we have $y_0^{**}+(1-p_I)y_1^{**}+(1-p_S)(1-p_I)y_2^{**}<0$ because $y_0^{**}<0$, $y_1^{**}<0$, $y_2^{**}<0$, which is inconsistent with our previous conclusion. Therefore, $y_0^{**}>0$, $y_1^{**}>0$ must hold.
		
		The remaining question is the sign of $y_2^{**}$. According to the second equation in Eq.~(\ref{system}), we have
		\begin{equation}\label{c5}
			y_0^{**}=\frac{\alpha}{r+\mu} x_0^{**} [y_0^{**}+(1-p_I)(y_1^{**}+y_2^{**})]/n^*.
		\end{equation}
		Since we have $y_0^{**}>0$, we know $x_0^{**} [y_0^{**}+(1-p_I)(y_1^{**}+y_2^{**})]>0$ as well. The sign of $x_0^{**}$ can be easily judged: if $x_0=0$, then $\dot{x}_0=\varepsilon_0 \Lambda>0$ so that $x_0<0$ never happens if the system starts evolving from a meaningful initial state where $x_0>0$. Therefore, $x_0^{**}>0$ and $y_0^{**}+(1-p_I)(y_1^{**}+y_2^{**})>0$ must hold as well. Then, $y_2^{**}>0$ is ensured by the third equation in Eq.~(\ref{c4}). 
		
		Therefore, $\mathcal{R}_0>1$ is a sufficient and necessary condition for $y_0^{**}>0$, $y_1^{**}>0$, $y_2^{**}>0$.
		
		To check if the solution of $y_0^{**}$, $y_1^{**}$, and $y_2^{**}$ really exists, we can study the existence of $y_0^{**}+(1-p_I)(y_1^{**}+y_2^{**})$. Then, the solution of $y_0^{**}$, $y_1^{**}$, and $y_2^{**}$ can be naturally obtained by Eq.~(\ref{c4}). For convenience, we denote $Y=y_0^{**}+(1-p_I)(y_1^{**}+y_2^{**})$. Multiplying the three equations of $y_0^{**}$, $y_1^{**}$, $y_2^{**}$ in Eq.~(\ref{c4}) by $1$, $1-p_I$, $1-p_I$, and adding them together, we get
		\begin{equation}\label{c6}
			Y
			=
			\dfrac{\dfrac{\alpha}{r+\mu}\varepsilon_0 \Lambda}{\alpha+\dfrac{\mu n^*}{Y}}
			+
			(1-p_I)\dfrac{\dfrac{\alpha}{r+\mu}\varepsilon_1 \Lambda}{\alpha+\dfrac{\mu n^*}{Y}}
			+
			(1-p_I)\dfrac{\dfrac{\alpha}{r+\mu}(1-p_S)\varepsilon_2 \Lambda}{\alpha(1-p_S)+\dfrac{\mu n^*}{Y}},
		\end{equation}
		which can be simplified as follows when $Y\neq 0$.
		\begin{equation}\label{c7}
			aY^2+bY+c=0,
		\end{equation}
		where
		\begin{equation*}
			\left\{\begin{aligned}
				a=&~\dfrac{\alpha}{\mu n^*}(1-p_S), \\
				b=&~ 1-p_S+1-\mathcal{R}_0+\frac{\alpha}{r+\mu}p_S(\varepsilon_0+(1-p_I)\varepsilon_1), \\
				c=&~\frac{\mu n^*}{\alpha}(1-\mathcal{R}_0). \\
			\end{aligned}\right.
		\end{equation*}
		We can see that $a>0$ always holds. The signs of $b$ and $c$, however, depend on $\mathcal{R}_0$. Then, the simple use of Vieta's theorem can judge the existence of $Y$. When $\mathcal{R}_0<1$, we have $c/a>0$ and $-b/a<0$; therefore, both roots are negative. When $\mathcal{R}_0>1$, we have $c/a<0$; therefore, one root is positive and the other is negative. The positive root is the unique solution of $Y$ and $y_0^{**}+(1-p_I)(y_1^{**}+y_2^{**})$. Then, Eq.~(\ref{c4}) can provide the unique solution of $y_0^{**}$, $y_1^{**}$, and $y_2^{**}$ (note that the negative root of $Y$ cannot lead to positive $y_0^{**}$, $y_1^{**}$, and $y_2^{**}$ because we have previously shown $y_0^{**}>0$, $y_1^{**}>0$, and $y_2^{**}>0$ if $\mathcal{R}_0>1$).

	\end{appendices}

\end{document}